\documentclass[12pt,a4paper]{article}

\makeatletter

\renewcommand\section{\@startsection
 {section}{1}{0pt}
  {-\baselineskip}
  {0.5\baselineskip}
{\raggedright\reset@font\normalsize\bfseries\mathversion{bold}}}

\renewcommand{\subsection}{\@startsection
{subsection}{2}{0mm}
  {-\baselineskip}
{0.5\baselineskip}
{\normalfont\normalsize\itshape}}

\renewcommand{\subsubsection}{\@startsection
{subsubsection}{3}{0mm}
  {-\baselineskip}
  {0.3\baselineskip}
{\normalfont\normalsize\itshape}}

\long\def\@makefntext#1{%
  \noindent \hb@xt@ 0.2em{\hss \@makefnmark}\hskip0.5em\relax#1}%

\makeatother

\usepackage{lmodern}
\usepackage{amsmath,amssymb,amsthm}
\usepackage{mathrsfs}
\usepackage[USenglish,british,american,australian]{babel}
\usepackage[ruled,linesnumbered]{algorithm2e}
\usepackage{float}
\usepackage{wrapfig}
\usepackage{enumerate}
\usepackage{graphicx,color}
\usepackage{natbib}
\bibpunct[:]{(}{)}{,}{a}{}{,}
\usepackage{setspace}
\usepackage[margin=1in]{geometry}
\usepackage{soul}
\usepackage[title]{appendix}

\newcommand{\Var}{\mathrm{Var}}

\newcommand{\E}{\mathbb E}

\theoremstyle{plain}
\newtheorem{theorem}{Theorem}

\newtheorem{lemma}{Lemma}
\newtheorem{proposition}{Proposition}

\theoremstyle{definition}
\newtheorem{definition}{Definition}
\newtheorem{example}[theorem]{Example}

\newtheorem{assumption}[theorem]{Assumption}

\newcounter{num}

\begin{document}
\bibliographystyle{plainnat}

\title{Conference Design for Consensus under KL Acceptance}
\author{
Liping Tang\footnote{Department of Philosophy, Institute of Logic and Cognition, Sun Yat-sen University, Guangzhou, 510275, P.R.China}
\quad Michiko Ogaku\footnote{Faculty of Economics, Nagasaki University, 4-2-1 Katafuchi, Nagasaki, 850-8506 Japan}
}
\date{
}
\maketitle


\begin{abstract}
We study consensus formation and information acquisition in conferences where DeGroot-type opinion updating is constrained by an acceptance rule based on Kullback–Leibler divergence. Players choose private information experiments before deliberation, and these choices determine both their posterior distributions and acceptance tolerances. We establish sufficient conditions for global consensus under specified conference structures, including repeated bilateral conferences organized by a fully receptive mediator. Under full reimbursement of the mediator’s conference costs, the information-acquisition game has a weighted potential and admits a pure-strategy Nash equilibrium. We develop a numerical procedure using verified payoff bounds over the continuous posterior domain to certify approximate equilibria and assess partial reimbursement policies that preserve a given equilibrium profile. The binary state, three-player examples illustrate how conference order interacts with state variance. At larger variances in the tested range, placing the demanding nonmediator last is associated with lower conference costs and lower certified sufficient planner reimbursement, but greater information acquisition expenditure by the mediator. These findings reveal a trade-off between the planner’s fiscal expenditure and the mediator’s private burden, and show that the conference structure matters economically.
\end{abstract}



\section{Introduction}
Communities often need to take collective action under uncertainty, for example, to manage a common resource or protect a shared environment, but their members may be unwilling to incorporate opinions that differ substantially from their own.
We study dynamic opinion formation in which player $i$, with belief $p_i$, accepts another belief $q$ only if 
\[D_{\mathrm{KL}}(q || p_i)\leq\eta_i,\] 
where $\eta_i>0$ may vary across players. Although an unconstrained DeGroot process on a complete network reaches consensus, the corresponding process with KL acceptance may become trapped in several opinion clusters because the realized influence network is endogenous to the evolving beliefs.

A planner commits to a signal-contingent conference protocol $\Gamma$ specifying conference membership, order, and stopping. We show that if one naive listener accepts every opinion throughout the continuation and all tolerances are positive, then a finite sequence of conferences achieves exact global consensus. The protocol repeatedly convenes the naive listener with each other player until the listener's belief is sufficiently close for that player to accept it. 
Complete cycles generate a contracting sequence of the opinion diameter. 
Once all beliefs lie in their acceptance regions, a final all-player conference produces exact consensus. 

We then endogenize acceptance tolerances. Before the state is realized, each player chooses an information strategy determining acceptance tolerance $\eta_i$. A smaller $\eta_i$ corresponds to a more informative, more costly experiment and a narrower acceptance set; each player incurs information cost measured by Shannon mutual information. The protocol again consists of bilateral conferences followed by a final all-player conference. Where a naive listener is a mediator who moves its own belief into the target player's acceptance region before both adopt their midpoint. 
Each conference period incurs a conference time cost. The mediator bears this cost. 
Before information acquisition, the planner commits to a reimbursement rule to deter the mediator's loads. 

At consensus, the community chooses an action $y$, and each player $i$ receives decision payoff 
$-r_i(y-X)^2$, where $X$ is the payoff-relevant state value and  $r_i>0$.  
The resulting dynamic game combines free riding in information acquisition with endogenous influence and delay. The planner designs $\Gamma$ to balance decision accuracy, information-acquisition incentives, and the time required for consensus.

\paragraph{Related work.} ~
DeGroot-type learning was developed as a framework for analysing the role of social influence in consensus formation. Previous studies have identified conditions under which consensus emerges \citep{degroot1974reaching}, conditions ensuring the accuracy of the resulting consensus \citep{golub2010naive}, and conditions for consensus when agents accept others’ opinions only within certain confidence bounds \citep{rainer2002opinion}.
Our study belongs to the literature on the model of \cite{rainer2002opinion}. 
A distinctive feature of our approach is an acceptance rule based on Kullback–Leibler (KL) divergence.

In cooperative game theory, conference structures and communication situations have been introduced, where meetings and communication opportunities among players are represented through hypergraphs\citep{myerson1980conference, van1992allocation}. The Myerson value and the position value have been proposed and characterized as fair allocation rules in these settings. These studies, however, address allocation problems under a given communication situation. The current study sheds light on the impact of exogenous conference structures on opinion dynamics.

Related literature studies information acquisition and social learning. \cite{risco2025network} examines costly private information acquisition before DeGroot-type social learning, while \cite{porcher2025migration} combines dynamic rational inattention with local log-linear belief aggregation in a migration model. Our study introduces conference protocol and endogenous acceptance tolerance into a static rational-inattention framework and examines their impacts on decision accuracy, information acquisition incentives, and the conference costs for coordination. 

\paragraph{Our contribution and paper organization.} ~
Our contribution is threefold. First, we identify sufficient conditions 
 for consensus when players are unwilling to incorporate opinions that differ substantially from their own, and when exogenously specified conference structures are introduced. Second, we show the existence of a Nash equilibrium for players' information strategy in this situation. 
Third, we provide numerical examples that suggest the bilateral conference cycle that places the demanding (narrow acceptance range) player last 
could be the most economical choice for the planner but puts a larger burden on the mediator to acquire information itself. 

The organization of the paper is as follows.   
Section \ref{sec:model} sets up the model and shows the sufficient condition for consensus. Section \ref{sec:information-acquisition} introduces the information acquisition game and the endogenous tolerance determined there.   Section \ref{sec:existence} shows existence of a pure-strategy Nash equilibrium for the information acquisition game. 
Section \ref{sec:computation} presents a procedure for computing and certifying approximate equilibria. 
Section \ref{sec:numerical-e} presents numerical examples, and  Section \ref{sec:future} considers future work and concludes. Proofs are provided in the Appendix.

\section{Model}\label{sec:model}
\paragraph{Preliminaries.} ~
We write \([n]=\{1,\ldots,n\}\) for the set of players. Player \(i\in[n]\) has belief \(p_i(t)\) at time \(t\)
and acceptance tolerance \(\eta_i>0\). Beliefs are defined on a measurable state space
\((\Omega,\mathcal F)\) and compared on a finite measurable
partition \(\mathcal E=\{E_1,\ldots,E_\tau\}\) of \(\Omega\). Initial beliefs have full support on \(\mathcal E\).

\paragraph{KL acceptance.} ~
Player \(i\), with tolerance \(\eta_i>0\), accepts an opinion \(q\) if and only if
\[
D_{\mathrm{KL},\mathcal E}(q\|p_i(t))
:=\sum_{k=1}^{\tau}q(E_k)
\log\frac{q(E_k)}{p_i(t)(E_k)}
\leq\eta_i.
\]
Thus, player \(i\)’s acceptance neighborhood at time $t$ is given by 
\[
N_i(t)
:=\bigl\{j\in[n]:
D_{\mathrm{KL},\mathcal E}(p_j(t)\|p_i(t))
\leq\eta_i\bigr\}.
\]
During conference \(C(t)\) at time $t$, participants average all members'
opinions, replacing each rejected opinion with their own belief:
\[
p_i(t+1)=
\begin{cases}
\displaystyle
p_i(t)+\frac{1}{|C(t)|}
\sum_{j\in C(t)\cap N_i(t)}
\bigl(p_j(t)-p_i(t)\bigr),
& i\in C(t),\\[3mm]
p_i(t),&i\notin C(t).
\end{cases}
\]
Nonparticipants keep their beliefs.

\begin{definition}[Naive listener] ~

We call $i \in [n]$ a naive listener if 
\[
D_{\mathrm{KL},\mathcal E}(p_j(t)\|p_i(t))\leq \eta_i \qquad j \in [n], t \geq 0.
\] 
\end{definition}

A naive listener may help to reach a global consensus, but it still depends on the conference structure. 
The following example shows the case where a three-player conference may fail despite the presence of a naive listener. 
\begin{example}[A three-player conference may fail despite a naive listener] Consider three players with binary beliefs
\[
p_1(0)=\left(\tfrac12,\tfrac12\right),\quad
p_2(0)=\left(\tfrac9{10},\tfrac1{10}\right),\quad
p_3(0)=\left(\tfrac1{10},\tfrac9{10}\right),
\]
tolerances \((\eta_1,\eta_2,\eta_3)=(3,0.1,0.1)\), and global conference \(C(t)=[3]\) at every period. Direct calculation gives \(N_1(0)=[3]\), \(N_2(0)=\{2\}\), and \(N_3(0)=\{3\}\). Players \(2\) and \(3\) therefore retain their beliefs. Player \(1\)'s belief also remains unchanged because it equals the average of all three beliefs. Thus the profile is stationary: player \(1\) is a naive listener throughout, yet global consensus is never reached.
\end{example}

Players with positive KL tolerances can reach a global consensus in the following conference structure in the presence of a naive listener.  

Fix a naive listener \(i^*\) and an order \(j_1,\ldots,j_{n-1}\) of the other players. The conference structure consists of repeated bilateral cycles followed by one grand conference:
\[
\underbrace{
\{i^*,j_1\}\;\longrightarrow\;\cdots\;\longrightarrow\;\{i^*,j_{n-1}\}
}_{\text{repeat cycles}}
\;\longrightarrow\;[n].
\]
First, they meet in pairs.
While rejected, the listener moves halfway toward the
selected player's unchanged belief.
Once accepted, both average their beliefs.
Second, they repeat the cycle. Continue until all players accept one another.
Finally, they meet all together, where everyone adopts the same average belief:
\[
p^*=\frac1n\sum_{j=1}^n p_j(T),
\]
where \(T\) is the time just before the final meeting.

\begin{proposition}\label{prop:naive-listener}
With a naive listener and positive KL tolerances,
this protocol reaches exact consensus in finite time.
\end{proposition}

\section{Information acquisition}\label{sec:information-acquisition}
Now we endogenise players' tolerance thresholds by introducing the information acquisition stage before the conferences. 
The timeline is as follows. 
\begin{enumerate}
\setlength\itemsep{3pt}
\item The planner commits to the retained share $\lambda$ and a conference protocol $\Gamma$, including a designated mediator $h$, a fixed cyclic order $(j_1,\ldots,j_{n-1})$ of the other players, the adjustment rule, and the global stopping rule.
\item All players simultaneously choose information strategies $\mu_i$. Each choice determines the distribution of player $i$'s initial posterior and the ex ante acceptance tolerance $\eta_i[\mu_i;\gamma_i]$. 
\item The state $J$ is realized, with payoff-relevant value $X=x_J$.
\item Private signals are generated, and players form their initial posteriors $\Pi_i$. \item The conference protocol $\Gamma$ is implemented.
\item At consensus, the collective action $y^\Gamma$ is taken, the
announced reimbursement is paid, and the payoffs are determined.
\end{enumerate}
\paragraph{Environment.} ~
Let \(\Theta=\{1,\ldots,\tau\}\) index the cells of \(\mathcal E\), with common prior \(p_{\mathcal E}\) satisfying \(p_k:=p_{\mathcal E}(k)>0\). The state is \(J\sim p_{\mathcal E}\), and the payoff-relevant quantity is \(X=x_J\). Fix \(0<\varepsilon<\min_k p_k\) and restrict beliefs to the compact  set 
\[
K_\varepsilon
:=\{\pi\in\Delta(\Theta):\pi(k)\geq\varepsilon
\text{ for all }k\}.
\]
Players simultaneously choose distributions of posteriors in $K_\varepsilon$, satisfying Bayes plausibility: 
\[
\mu_i\in\mathcal D_i:=\mathcal M_\varepsilon(p_{\mathcal E})
:=\left\{\mu\in\Delta(K_\varepsilon):
\int_{K_\varepsilon}\pi\,\mu(d\pi)=p_{\mathcal E}\right\}.
\]
We assume that posteriors are independent conditional on $J$ given $\mu=(\mu_1,\ldots,\mu_n)$ with 
\[
\Pr_\mu(\Pi_i\in d\pi_i\mid J=k)
=\frac{\pi_i(k)}{p_k}\,\mu_i(d\pi_i).
\]
Their realizations form the initial conference opinions. Write \(\nu_\mu\) for the induced distribution of \(\Pi=(\Pi_1,\ldots,\Pi_n)\), and \(\mathbb E_\mu\) for expectation under the joint law of \(J\) and \(\Pi\).

\paragraph{Endogenous acceptance tolerance.}~
For an experiment \(\mu_i\), define player \(i\)’s residual uncertainty by
\[
R_i(\mu_i):=\int_{K_\varepsilon}\operatorname{Var}_{\pi}(X)\,\mu_i(d\pi).
\]
Given a player-specific parameter \(\gamma_i>0\), the experiment induces acceptance tolerance
\[
\eta_i[\mu_i;\gamma_i]
:=H\bigl(R_i(\mu_i);\gamma_i\bigr),
\]
where \(H(\cdot;\gamma_i)\) is nondecreasing. 
Thus, lower residual uncertainty leads to a weakly smaller tolerance for disagreement. 
The tolerance is determined before the private signal is observed and remains fixed throughout the conferences.

We assume that all players' tolerances are positive. 
\begin{assumption}[Positive endogenous tolerances]\label{ass:positive-endogenous-tolerance}
For every $i$, the function $H(\cdot;\gamma_i)$ is continuous on
$[0,\Var_{p_{\mathcal E}}(X)]$ and strictly positive on the
attainable residual-risk set
\[
\mathcal Z_i(p_{\mathcal E})
:=\{R_i(\mu_i):\mu_i\in\mathcal M_\varepsilon(p_{\mathcal E})\}.
\]
\end{assumption}
Since $\mathcal{D}_i$ is nonempty, weak$^\ast$ compact, and convex (Lemma~\ref{lem:admissible-experiment}), and $R_i$ is weak$^\ast$ continuous and linear, 
the attainable set $\mathcal{Z}_i(p_\mathcal{E})=R_i(\mathcal{D}_i)$ is a nonempty compact interval. 

\paragraph{The mediator's adjustment.}~
Fix a mediator $h$. We assume the mediator $h$'s uniform tolerance. 

\begin{assumption}[Uniform mediator tolerance]\label{ass:mediator-tolerance}
For every admissible experiment $\mu$,
\[
\eta_h[\mu_h;\gamma_h]\geq d(\Pi)
\quad\nu_\mu\text{-almost surely}, \quad d(\pi):=\max_{a,b} D_{KL}(\pi_a \| \pi_b). 
\]
\end{assumption}

In one bilateral conference, the mediator moves toward the target only as far as needed for acceptance, and both then average.
For mediator and target beliefs \(u,v\in K_\varepsilon\), define
\[
\begin{aligned}
\theta_i(u,v;\eta_i)
&:=\max\left\{\theta\in[0,1]:
D_{\mathrm{KL}}\bigl(v+\theta(u-v)\|v\bigr)\leq\eta_i
\right\}, \\ \vspace{3pt}
F_i(u,v;\eta_i)
&:=v+\frac{\theta_i(u,v;\eta_i)}{2}(u-v).
\end{aligned}
\]
If the target already accepts \(u\), then \(\theta_i=1\)
and \(F_i=(u+v)/2\).
All other beliefs remain unchanged.

\paragraph{Global consensus.}~
Starting from \(q^0=\pi\), let \(q^r\) denote the profile after \(r\) completed bilateral stages. 
With positive tolerances, repeated cycles achieve mutual
acceptance after \(\ell=\ell^\Gamma(\pi;\eta)<\infty\)
bilateral stages. A final global conference gives the common belief and collective action 
\[
p^{*,\Gamma}(\pi;\eta)
=\frac1n\sum_{j=1}^n q_j^{\ell^\Gamma(\pi;\eta)},
\qquad
y^\Gamma(\pi;\eta)=\mathbb E_{p^{*,\Gamma}(\pi;\eta)}[X], 
\] where the stopping stage number $\ell^\Gamma(\pi;\eta)$ is formally defined as follows. 
\[
\ell^\Gamma(\pi;\eta)
:=\inf\left\{r\in\mathbb N_0:
\max_{a,b\in[n]}D_{\mathrm{KL}}(q_a^r\|q_b^r)
\leq\min_{i\in[n]}\eta_i\right\}.
\]
We obtain the following regularity. 
\begin{lemma}\label{lem:mediated-consensus-continuity}
For every positive tolerance vector, the protocol reaches
global consensus after finitely many bilateral stages,
uniformly over $\pi\in K_\varepsilon^n$. The consensus
$p^{*,\Gamma}$ and action $y^\Gamma$ are jointly continuous in
$(\pi,\eta)\in K_\varepsilon^n\times(0,\infty)^n$ for the fixed
cyclic order.
\end{lemma}

\paragraph{Cumulative conference cost.} ~
The planner commits to the mediator's retained share
\(\lambda\in[0,1]\) before experiments are chosen.

For bilateral stage \(r\), with target \(j_r\), write
\[
\theta_r^\Gamma
:=\theta_{j_r}(q_h^{r-1},q_{j_r}^{r-1};\eta_{j_r}).
\]
Cumulative conference cost is the final-conference charge
plus the bilateral stage durations:
\[
B_h^\Gamma(\pi;\eta)
:=1+\sum_{r=1}^{\ell^\Gamma(\pi;\eta)}
(2-\theta_r^\Gamma).
\]
The leading \(1\) is a fixed charge for the final global conference. 
We assign each bilateral stage the continuous duration index \(2-\theta_r^\Gamma\): one unit for agreement and \(1-\theta_r^\Gamma\) units for the mediator’s adjustment. 
The sum includes all completed bilateral stages before stopping, across all cycles. The mediator bears a share \(\lambda\in[0,1]\) of this cost, while the planner reimburses the remaining share \(1-\lambda\); nonmediators incur no conference cost.

Importantly, $B_h^\Gamma$ is jointly lower semicontinuous and has a finite common bound. 
\begin{lemma}\label{lem:stopped-cost-regularity}
For the fixed cyclic order, $B_h^\Gamma$ is jointly lower
semicontinuous on $K_\varepsilon^n\times(0,\infty)^n$ and
locally bounded in the tolerance vector, uniformly over
initial profiles. It satisfies
\[
1+\ell^\Gamma(\pi;\eta)
\leq B_h^\Gamma(\pi;\eta)
\leq1+2\ell^\Gamma(\pi;\eta).
\]
Under Assumption~\ref{ass:positive-endogenous-tolerance},
there is a finite common bound $M$ on the number of bilateral
stages, across all admissible experiment profiles and initial
posteriors. Consequently,
\[
1\leq B_h^\Gamma(\pi;\eta[\mu;\gamma])\leq1+2M.
\]
\end{lemma}

\paragraph{Ex ante and interim utilities} ~
Fix \(\Gamma\), \(\gamma\), and \(\lambda\), and write \(\eta=\eta[\mu;\gamma]\). After observing posterior \(\pi_i\), player \(i\)’s interim decision payoff and conditional conference cost are given by 
\[
\begin{aligned}
V_i^\Gamma(\pi_i;\mu_{-i},\eta)
&:=-r_i\,\mathbb E_\mu\!\left[
\bigl(y^\Gamma(\Pi;\eta)-X\bigr)^2
\mid\Pi_i=\pi_i
\right],\\
C_i(\pi_i;\mu_{-i},\eta)
&:=\mathbb E_\mu\!\left[
B_h^\Gamma(\Pi;\eta)\mid\Pi_i=\pi_i
\right],
\end{aligned}
\]
where \(r_i>0\) weights decision loss. For every \(\rho\in K_\varepsilon\),
including posteriors outside \(\operatorname{supp}(\mu_i)\), conditional
expectations use the version
\begin{equation}\label{eq:posterior-conditional-law}
\Pr_\mu(J=k,\Pi_{-i}\in d\pi_{-i}\mid\Pi_i=\rho)
:=\rho(k)\bigotimes_{j\ne i}
\left[\frac{\pi_j(k)}{p_k}\,\mu_j(d\pi_j)\right].
\end{equation}
This law depends only on \(\rho\) and \(\mu_{-i}\).
Information acquisition costs are
\[
c_i(\mu_i)
:=\kappa_i\int_{K_\varepsilon}
D_{\mathrm{KL}}(\pi_i\|p_{\mathcal E})\,\mu_i(d\pi_i),
\qquad \kappa_i>0.
\]
Define the net payoff integrand
\[
\begin{aligned}
f_i^\lambda(\pi_i;\mu_{-i},\eta)
:={}&V_i^\Gamma(\pi_i;\mu_{-i},\eta)
-\kappa_iD_{\mathrm{KL}}(\pi_i\|p_{\mathcal E})-\lambda\mathbf1_{\{i=h\}}
C_i(\pi_i;\mu_{-i},\eta).
\end{aligned}
\]
By the law of iterated expectations, ex ante utility is given by 
\[
U_i^\lambda(\mu;\Gamma)
=\int_{K_\varepsilon}
f_i^\lambda(\pi_i;\mu_{-i},\eta[\mu;\gamma])\,\mu_i(d\pi_i).
\]
We define the information acquisition game as follows. 
\begin{definition} 
For fixed primitives and conference protocol \(\Gamma\) with mediator \(h\) with retained share $\lambda$, we write the information-acquisition game 
\[
\mathcal G^\lambda(\Gamma,\gamma)
=\left([n],(\mathcal D_i)_{i\in[n]},(U_i^\lambda)_{i\in[n]}\right), 
\] where 
players simultaneously choose experiments \(\mu_i\) before signals are realised. 
\end{definition}

\section{Existence Results}\label{sec:existence}

\paragraph{Regularity.} ~ 
In the following lemma, we establish continuity of \(U_i^\lambda\) when \(i\neq h\) or \(\lambda=0\), and upper semicontinuity for the mediator when \(\lambda>0\). This ensures the weighted potential function $\Phi$ defined after the lemma is continuous at full reimbursement. 

\begin{lemma}\label{lem:acquisition-payoff-regularity}
Let $E=C(K_\varepsilon)$, $E^*=\mathfrak B(K_\varepsilon)$, and $\mathcal D=\prod_{j=1}^n\mathcal D_j$,
where 
\begin{equation}
\mathcal D_i
=
\mathcal M_\varepsilon(p_{\mathcal E})
=
\left\{\mu_i\in\Delta(K_\varepsilon):
\int_{K_\varepsilon}\pi\,\mu_i(d\pi)=p_{\mathcal E}
\right\}, \nonumber
\end{equation}
$E$ has the supremum norm and $\mathfrak B(K_\varepsilon)$
is the space of finite signed Borel measures, paired by
\[
\langle g,\zeta\rangle
:=\int_{K_\varepsilon}g\,d\zeta.
\]
Under Assumption~\ref{ass:positive-endogenous-tolerance},
$U_i^\lambda:\mathcal D\to\mathbb R$ is continuous when
$i\ne h$ or $\lambda=0$, and upper semicontinuous for the
mediator when $\lambda>0$, in the product of the weak$^*$
topologies $\sigma(E^*,E)$. In particular, for every fixed share,
\[
\mu_j^m\overset{*}{\rightharpoonup}\mu_j
\quad\text{for every }j\in[n]
\quad\Longrightarrow\quad
\limsup_{m\to\infty}U_i^\lambda(\mu^m;\Gamma)
\leq U_i^\lambda(\mu;\Gamma).
\]
At fixed positive tolerances, the same continuity and upper
semicontinuity statements hold for $f_i^\lambda$ as a
function of the player's posterior.
\end{lemma}

For opponents' experiments $\mu_{-i}$, define each player $i$'s best-response correspondence 
\begin{equation}
\operatorname{BR}_i^\Gamma(\mu_{-i};\gamma,\lambda)
:=
\operatorname*{arg\,max}_{\widehat\mu_i\in\mathcal D_i}
U_i^\lambda(\widehat\mu_i,\mu_{-i};\Gamma). \nonumber 
\end{equation}
Since $\mathcal{D}$ is compact and nonempty and each $U_i$ is upper semicontinuous, we see that 
$\operatorname{BR}_i^\Gamma(\mu_{-i};\gamma,\lambda)$ is nonempty and compact. 

\paragraph{Weighted-potential property.} ~
Under full reimbursement, $\lambda=0$, and fixed $\Gamma$ and
$\gamma$, the information-acquisition game $\mathcal G^0(\Gamma,\gamma)$ is a
weighted potential game \citep{monderer1996potential} with
\[
\Phi(\mu)
:=-L(\mu)-\sum_{j=1}^n\frac{c_j(\mu_j)}{r_j},
\qquad
\mu\in\mathcal D:=\prod_{j=1}^n\mathcal D_j,
\]
where $L(\mu)$ is the expected squared decision error:
\[
\begin{aligned}
L(\mu)
&:=\mathbb E_\mu\!\left[
\bigl(y^\Gamma(\Pi;\eta[\mu;\gamma])-X\bigr)^2
\right].
\end{aligned}
\]
Since \(U_i^0(\mu;\Gamma)=-r_iL(\mu)-c_i(\mu_i)\), and opponents’ acquisition costs remain unchanged under unilateral deviations, \(\Phi\) is a weighted potential.
Since \(\Phi\) is continuous on the nonempty weak\(^*\) compact set \(\mathcal D\), it attains a maximum at some \(\mu^*\). 
For every player \(i\) and for any unilateral deviation \(\widehat\mu_i\in\mathcal D_i\),
\[
\begin{aligned}
&U_i^0(\widehat\mu_i,\mu_{-i}^*;\Gamma)-U_i^0(\mu^*;\Gamma) =r_i\left[
\Phi(\widehat\mu_i,\mu_{-i}^*)-\Phi(\mu^*)
\right]\leq0.
\end{aligned}
\]
Thus every global maximizer of \(\Phi\) is a pure-strategy Nash equilibrium, $\mu_i^\ast \in \operatorname{BR}_i^\Gamma(\mu_{-i}^\ast;\gamma,0)$ for all $i$, establishing existence.
\begin{proposition}\label{prop:acquisition-equilibrium-existence}
Under Assumption~\ref{ass:positive-endogenous-tolerance},
the simultaneous information-acquisition game with full
reimbursement, $\mathcal G^0(\Gamma,\gamma)
$, has a pure-strategy Nash equilibrium.
Every global maximizer of $\Phi$ on $\mathcal D$ is such an
equilibrium. Heterogeneous positive weights $r_i$ are allowed.
\end{proposition}

In the benchmark where every player bears the same retained
conference cost and $r_i=r>0$ for all $i$,
the function
\[
\Phi^\lambda(\mu)
:=-rL(\mu)-\lambda \mathbb{E}_\mu[B_h^\Gamma(\Pi; \eta[\mu; \gamma])]-\sum_{j=1}^n c_j(\mu_j)
\]
is an exact potential. Its upper semicontinuity and the
compactness of the strategy space imply the existence of a
pure-strategy Nash equilibrium.

When only the mediator bears the retained conference cost,
this potential argument does not generally apply to every
player for $\lambda>0$. Nevertheless, sufficiently high
information-acquisition costs guarantee an uninformative
equilibrium, as the following proposition shows.

\begin{proposition}[Existence with sufficiently costly information]\label{prop:existence-costly-information-case} 
Define
\[
\Delta_x:=\max_{k\in\Theta}x_k-\min_{k\in\Theta}x_k.
\]
Suppose that
\[
\kappa_i\geq \frac{r_i\Delta_x^2}{2}
\qquad\text{for every }i\in[n].
\]
Then the uninformative profile
\[
\mu^0=(\delta_{p_{\mathcal E}},\ldots,\delta_{p_{\mathcal E}})
\]
is a pure-strategy Nash equilibrium for every retained share \(\lambda\in[0,1]\) in 
$\mathcal G^\lambda(\Gamma,\gamma)$. In particular, an equilibrium exists with zero reimbursement, \(\lambda=1\).
\end{proposition}

\paragraph{Partial Reimbursement.} ~
Let $\mu^*$ be a Nash equilibrium under full reimbursement.
For each mediator deviation $b_h\in\mathcal D_h$, define
\[
\begin{aligned}
A(b_h)
&:=U_h^0(\mu^*;\Gamma)
-U_h^0(b_h,\mu_{-h}^*;\Gamma)\geq0,\\
S(b_h)
&:=\mathbb E_{\mu^*}[B_h^\Gamma]
-\mathbb E_{(b_h,\mu_{-h}^*)}[B_h^\Gamma].
\end{aligned}
\]
Here $A$ is the baseline payoff sacrifice and $S$ is the expected conference-cost saving. 

At retained share \(\lambda\), the deviation gain is
\(-A(b_h)+\lambda S(b_h)\). 
Since nonmediators’ incentives are unchanged, the same profile \(\mu^*\) remains a Nash equilibrium if and only if
\[
0\leq\lambda\leq
\lambda_{\mathrm{NE}}^*(\mu^*)
:=\min\left\{1,
\inf_{\substack{b_h\in\mathcal D_h\\S(b_h)>0}}
\frac{A(b_h)}{S(b_h)}
\right\},
\]
with \(\inf\varnothing=+\infty\). 

\section{Computational Aspects}\label{sec:computation}
We formulate conditional best-response problems as conic linear programs over posterior distributions. 
We first establish strong duality. 
We then show how to compute approximate best responses and certify \(\xi\)-Nash equilibria under full reimbursement. 
Finally, we characterise and certify retained conference cost shares under which the same profile remains a \(\xi\)-Nash equilibrium. 
In the following, fix opponents’ experiments \(\mu_{-i}\) and an attainable residual-risk level \(z_i\). Set \(\bar\eta_i=H(z_i;\gamma_i)\), and write
\(f_i(\pi):=f_i^\lambda(\pi;\mu_{-i},\bar\eta)\).
Let \(\mathcal C_i:=\operatorname{cone}(\Delta(K_\varepsilon))\). 
\paragraph{Conic linear optimization problem.} ~
The conditional best-response problem is written as the following moment problem 
\[
\begin{aligned}
(P_i)\qquad W_i(z_i;\mu_{-i}):=
\sup_{\mu_i\in\mathcal C_i}\quad&
\int_{K_\varepsilon}f_i(\pi)\,\mu_i(d\pi)\\
\text{subject to}&\quad
\int_{K_\varepsilon}\pi\,\mu_i(d\pi)=p_{\mathcal E}, \quad 
\int_{K_\varepsilon}\operatorname{Var}_\pi(X)\,\mu_i(d\pi)=z_i.
\end{aligned}
\]
The constraints impose Bayes plausibility and consistency with the prescribed residual risk. The first constraint also implies that \(\mu_i\) has total mass one. Both the objective and constraints are linear in \(\mu_i\).

The corresponding dual problem is given by 
\[
\begin{aligned}
(D_i)\qquad
\inf_{\alpha\in\mathbb R^\tau,\;\beta\in\mathbb R}\quad&
\alpha\cdot p_{\mathcal E}+\beta z_i\\
\text{subject to}\quad&
\alpha\cdot\pi+\beta\operatorname{Var}_\pi(X)
\geq f_i(\pi)
\quad\forall\pi\in K_\varepsilon.
\end{aligned}
\]
The dual constructs an upper bound on the payoff using the moment functions. 
\paragraph{Strong duality.}
Denote the optimal values of $(P_i)$ and $(D_i)$ by $\operatorname{val}(P_i)$ and $\operatorname{val}(D_i)$, respectively. 
We see below that the primal maximum is attained, and strong duality holds:

\begin{proposition}[Attainment and strong duality]\label{prop:strong-duality}
Let $z_i$ be attainable and let $f_i$ be bounded and upper
semicontinuous on the compact posterior domain
$K_\varepsilon$. Then $(P_i)$ attains its maximum and $\operatorname{val}(P_i)=\operatorname{val}(D_i)$.
\end{proposition}
\paragraph{Computation of best response.} ~
The moment-constraint argument of \cite{pinelis2016extreme} gives an optimizer of the form 
\[
\mu^\star_i = \sum_{a=1}^{s} w_a \delta_{\pi^a}, \quad s \leq \tau+1, \quad w_a>0, \quad \sum w_a=1,
\]
where both $w_a$ and $\pi^a \in K_\varepsilon$ are chosen by optimization. We use this property and compute an approximation by adaptively generating posterior locations. This is an adaptation of the column-generation framework of \cite{mehrotra2013generating}.  

For a general nonmediator, write the fixed attainable residual risk $z_i$ as \(z\) and 
write
\[
v(\pi):=\operatorname{Var}_\pi(X),
\qquad
g_z(\pi):=f_i^\lambda(\pi;\mu_{-i},\bar\eta(z)).
\]

The initial pool $Q_0:=\{p_{\mathcal E}\}\cup
\operatorname{vert}(K_\varepsilon)$ supports every attainable
residual risk. Indeed, every admissible posterior is a convex combination of the vertices $q^k \in K_\varepsilon$ and 
the attainable residual-risk set $\mathcal{Z}_i(p_\mathcal{E})$ is a compact interval. 
Write the prior as a convex combination of the
vertices, 
\[p_{\mathcal E}=\sum_k a_k q^k, \quad a_k:=\frac{p_k-\varepsilon}{1-\tau \varepsilon},\] and set
$\mu^{\mathrm V}:=\sum_k a_k\delta_{q^k}$.
Expressing each posterior as a convex combination of vertices
and using concavity of $v$ shows that $\mu^{\mathrm V}$
minimizes expected posterior variance; Jensen's inequality
shows that $\delta_{p_{\mathcal E}}$ maximizes it.
Thus, writing
\[
z_{\min}:=\sum_k a_kv(q^k),
\qquad z_{\max}:=v(p_{\mathcal E}),
\]
we have
\[
\mu_t:=(1-t)\mu^{\mathrm V}+t\delta_{p_{\mathcal E}},
\qquad
R_i(\mu_t)=(1-t)z_{\min}+tz_{\max},
\qquad t\in[0,1].
\]
These experiments are Bayes plausible and supported on $Q_0$,
so $(P_Q)$ defined below is feasible for every attainable $z$ whenever
$Q\supseteq Q_0$.

Enlarge the candidate pool \(Q\) by adding consensus beliefs
obtained by combinations
of player \(i\)'s current candidate posteriors and posteriors
in the supports of the opponents' experiments.

For the enlarged pool
\(Q=\{\pi^1,\ldots,\pi^m\}\supseteq Q_0\),
solve the restricted linear program
\[
(P_Q)\qquad L_Q(z):=
\max_{w_a\ge0}
\left\{
\sum_a w_ag_z(\pi^a):
\sum_a w_a\pi^a=p_{\mathcal E},\quad
\sum_a w_av(\pi^a)=z
\right\}.
\]
The optimal weights define a feasible experiment
with residual risk \(z\).

Using the optimal dual multipliers \((\alpha_Q,\beta_Q)\), search for posteriors with positive residual
\[
r_{z,Q}(\pi)
:=g_z(\pi)-\alpha_Q\cdot\pi-\beta_Qv(\pi).
\]
Add such posteriors to \(Q\) and reoptimize the weights. We call this process pricing. 
A verified bound
\[
\max_{\pi\in K_\varepsilon}r_{z,Q}(\pi)\le\delta
\]
certifies 
\[
L_Q(z)\le W_i(z;\mu_{-i})\le L_Q(z)+\delta.
\]
Indeed, by finite-dimensional LP strong duality, $L_Q(z)=\alpha_Q\cdot p_{\mathcal E}+\beta_Qz$.
Hence the verified global residual bound implies,
for every feasible experiment with residual risk $z$,
\[
\int_{K_\varepsilon}g_z(\pi)\,\mu_i(d\pi)
\le \alpha_Q\cdot p_{\mathcal E}+\beta_Qz+\delta
=L_Q(z)+\delta.
\]

A full best response additionally requires optimization over all attainable risks:
\[
W_i^\star(\mu_{-i})
:=\sup_{z\in\mathcal Z_i(p_{\mathcal E})}W_i(z;\mu_{-i}).
\]
Retain the best feasible experiment \(\widehat\mu_i\) found and refine verified upper bounds over a partition of the risk interval. 
 
Write $u_i:=U_i^\lambda(\tilde\mu^k;\Gamma)$ and  $L_i:=U_i^\lambda(\widehat\mu_i,\tilde\mu_{-i}^k;\Gamma)$,
where \(\widehat\mu_i\) is the best verified feasible candidate retained. 
We maintain conservative computationally verified bounds
\[
\underline u_i\le u_i\le\overline u_i,
\qquad
\underline L_i\le L_i,
\qquad
W_i^\star(\tilde\mu_{-i}^k)\le\overline W_i^{\mathrm{num}}, 
\] and for $I \subset \mathcal Z_i(p_{\mathcal E})=[z_{\mathrm{min}},z_{\mathrm{max}}]$, $\overline B_i^{\mathrm{num}}(I)
\geq \sup_{z\in I}W_i(z;\tilde{\mu}_{-i}^k)$.
Lemma~\ref{lem:interval-payoff-bound} provides a construction of these interval bounds and the resulting global upper bound.
Fix $\xi>0$ and parameter $1\le\iota<\infty$.
Player $i$ reports a verified improvement if
\[
\underline L_i-\overline u_i>\frac{\xi}{2\iota},
\]
and a certified deviation bound if
\[
\overline W_i^{\mathrm{num}}-\underline u_i\le\xi.
\]
Both reports may hold simultaneously.

Concerning the mediator's best response, under Assumption~\ref{ass:mediator-tolerance}, the mediator's
tolerance does not affect the decision, information, and conference costs. Thus, for fixed $\mu_{-h}$
and $\lambda$, the conditional net-payoff integrand is
independent of $\mu_h$, which enters ex ante utility only
through integration. The mediator's best-response problem
therefore optimizes directly over Bayes-plausible experiments,
without a residual-risk consistency constraint or an outer
risk search.

\paragraph{Nonmediator subroutine $\mathcal{A}_i$.} ~

To balance pricing with refinement of verified global upper
bounds, the algorithm allocates at most
$T_i^{\mathrm{pricing}}>0$ seconds to each pricing call at a
selected residual-risk level $z$. If pricing remains unfinished
when this limit is reached, the algorithm retains the verified
results and proceeds to global-bound refinement or another
risk level. The unfinished pricing search may be resumed later.
Reaching the pricing time limit does not itself provide an
optimality or equilibrium certificate.

The pseudocode below omits these time-management details.

\vspace{6pt}
{\small 
\begin{algorithm}[H]
\caption{Nonmediator subroutine $\mathcal A_i$: general branch}
\label{alg:nonmediator}

\KwIn{Player $i\ne h$, fixed profile $\tilde\mu^k$,
posterior pool $Q_i\supseteq Q_0$, accuracy $\xi>0$,
and parameter $1\le\iota<\infty$.}
\KwOut{A feasible candidate with improvement and
certification flags or an inconclusive result if computation is stopped.}

$\widehat\mu_i\gets\tilde\mu_i^k$\;
Initialize verified bounds
$\underline u_i\le u_i\le\overline u_i$ and
$\underline L_i\gets\underline u_i$\;
$\mathcal P_i\gets\{[z_{\min},z_{\max}]\}$,
$\overline W_i^{\mathrm{num}}\gets0$\;

\Repeat{$\mathsf{improving}\lor\mathsf{certified}$}{
    Select or revisit an attainable risk
    $z\in I\in\mathcal P_i$\;

    Set $\bar\eta_i=H(z;\gamma_i)$\;
$S_i\gets Q_i$\;
$T_i^{\mathrm{new}}\gets
\left\{
p^{*,\Gamma}\bigl((\rho,\pi_{-i});\bar\eta(z)\bigr)
\;\middle|\;
\rho\in S_i,\quad
\pi_{-i}\in
\prod_{j\ne i}\operatorname{supp}(\tilde\mu_j^k)
\right\}$\;
$Q_i\gets Q_i\cup T_i^{\mathrm{new}}$\;

    Refresh $g_z(\pi^a)$ for $\pi^a\in Q_i$\;

    Advance the restricted-LP and pricing searches,
    adding verified positive-residual points to $Q_i$\;

    Verify proposed experiments and retain the best
    feasible candidate $\widehat\mu_i$ and its
    payoff lower bound $\underline L_i$\;

    Refine $\mathcal P_i$ and the valid interval bounds
    $\overline B_i^{\mathrm{num}}(I)$
    over all attainable risks and posteriors\;

    $\overline W_i^{\mathrm{num}}\gets
    \max_{I\in\mathcal P_i}
    \overline B_i^{\mathrm{num}}(I)$\;
    Refine the payoff enclosures\;

    $\mathsf{improving}\gets
    [\,\underline L_i-\overline u_i>\xi/(2\iota)\,]$\;
    $\mathsf{certified}\gets
    [\,\overline W_i^{\mathrm{num}}-\underline u_i
       \le\xi\,]$\;
}
\Return{$(\widehat\mu_i,
\mathsf{improving},\mathsf{certified})$}\;
\end{algorithm}
}

\paragraph{Certification of $\xi$-Nash equilibria.} ~
We certify $\xi$-Nash equilibrium at full reimbursement.  
Fix \(\lambda=0\), \(\xi>0\), and \(1\le\iota<\infty\), the algorithm starts from a feasible profile \(\tilde\mu^0\). At round \(k\), it keeps \(\tilde\mu^k\) fixed while calling the player subroutines. It stops if every player has a verified certificate
\[
W_i^\star(\tilde\mu_{-i}^k)-U_i^0(\tilde\mu^k;\Gamma)
\le \overline W_i^{\mathrm{num}}-\underline u_i
\le\xi.
\]
The resulting profile is therefore a \(\xi\)-Nash equilibrium.

Otherwise, the algorithm selects one player \(i_k\) with
\(\underline L_{i_k}-\overline u_{i_k}>\xi/(2\iota)\)
and implements only that player’s proposal:
\[
\tilde\mu^{k+1}
=(\widehat\mu_{i_k}^k,\tilde\mu_{-i_k}^k).
\]
It refreshes affected payoff bounds and certificates at the new profile. If no verified improvement is available and some certificates remain unresolved, it continues the unfinished calls; stopping at this stage gives an inconclusive result.

Under full reimbursement, \(\lambda=0\), the weighted-potential property implies
\[
\Phi(\tilde\mu^{k+1})-\Phi(\tilde\mu^k)
>\frac{\xi}{2\iota r_{\max}},
\qquad r_{\max}:=\max_i r_i.
\]
Hence, for any finite global upper bound \(\overline\Phi\), the number \(N\) of accepted updates satisfies
\[
N\le
\frac{2\iota r_{\max}}{\xi}
\bigl[\overline\Phi-\Phi(\tilde\mu^0)\bigr].
\]
Thus the procedure reaches a \(\xi\)-Nash equilibrium provided every subroutine call terminates with a verified improvement or certificate. For \(\lambda>0\), the same equilibrium certificate remains valid, but the potential argument does not establish finite termination.

\paragraph{Approximate equilibria with $\lambda>0$ and certificates.} ~
We also compute how much conference cost the mediator can retain while preserving the same profile as a $\xi$-Nash equilibrium.
For a fixed retained share \(\lambda\), define player \(i\)’s unilateral regret at a feasible profile \(\mu\) by
\[
\operatorname{Reg}_i^\lambda(\mu)
:=\sup_{\nu_i\in\mathcal D_i}
U_i^\lambda(\nu_i,\mu_{-i};\Gamma)
-U_i^\lambda(\mu;\Gamma).
\]

To seek the largest certified $\lambda$ preserving $\xi$-Nash equilibrium, fix a profile \(\mu^0\) satisfying
\begin{align}
\operatorname{Reg}_i^0(\mu^0)\le\xi_0\le\xi
\qquad\text{for every }i. \label{ineq:regret}
\end{align}
Keeping $\mu^0$ fixed, the algorithm tests whether the mediator's regret satisfies 
$\operatorname{Reg}_h^\lambda(\mu^0)\leq\xi$
at different retained shares $\lambda\in[0,1]$.
Since nonmediators' payoffs are independent of $\lambda$,
their regret satisfies \eqref{ineq:regret}.

Each trial is certified if a verified upper bound of the mediator's regret 
$\overline{\operatorname{Reg}}_h(\lambda_m)\leq\xi$,
rejected if a verified lower bound $\underline\Delta_m$
on the payoff gain from a feasible mediator deviation
satisfies $\underline\Delta_m>\xi$, and inconclusive otherwise.
Let \(\lambda_A\) denote the largest certified trial share.

\section{Numerical Example}\label{sec:numerical-e}
For computational tractability, our numerical example focuses on the binary-state, three-player case, where prior $p_\mathcal{E}(x_1)=1/2$ and $\varepsilon=0.0005$. 
We consider two equally likely states with payoff values
\(x_1=\tfrac12-\sqrt{s}\) and \(x_2=\tfrac12+\sqrt{s}\),
so that \(\mathbb E[X]=\tfrac12\) and
\(\operatorname{Var}(X)=s\).

The panels in Figure \ref{fig:combined}  show the weighted decision cost, total information acquisition cost, conference cost at $\xi$-Nash equilibrium, and certified sufficient planner reimbursement using $\lambda_A$. The certification tolerance is $\xi=0.05$, except for case A defined below at variance $s=0.10$, where $\xi=0.067$. 
The horizontal axis is variance, so moving right means greater state uncertainty. 

Player 1 is the mediator, and meets player 2 and then player 3. 
A larger $\gamma$ means a less tolerant player. 
In blue, A places the demanding nonmediator last, $\gamma=(0.1,0.1,10)$. In orange, B places  
the demanding nonmediator first, $\gamma=(0.1,10,0.1)$. In green, C has two demanding nonmediators, $\gamma=(0.1,10,10)$. 
For the parameter range considered, $\gamma_i=0.1$ satisfies the naive-listener condition. 
Other parameters are: $r_i=10$ for all $i \in [3]$, $\kappa_i=0.1$ for all $i\in [3]$, and $H(z;\gamma_i)=e^z/\gamma_i$.
Case C at $s=1.0$ is omitted because a $0.05$-Nash equilibrium certificate was not obtained. 

\begin{figure}[H]
\begin{center}
\includegraphics[width=16cm]{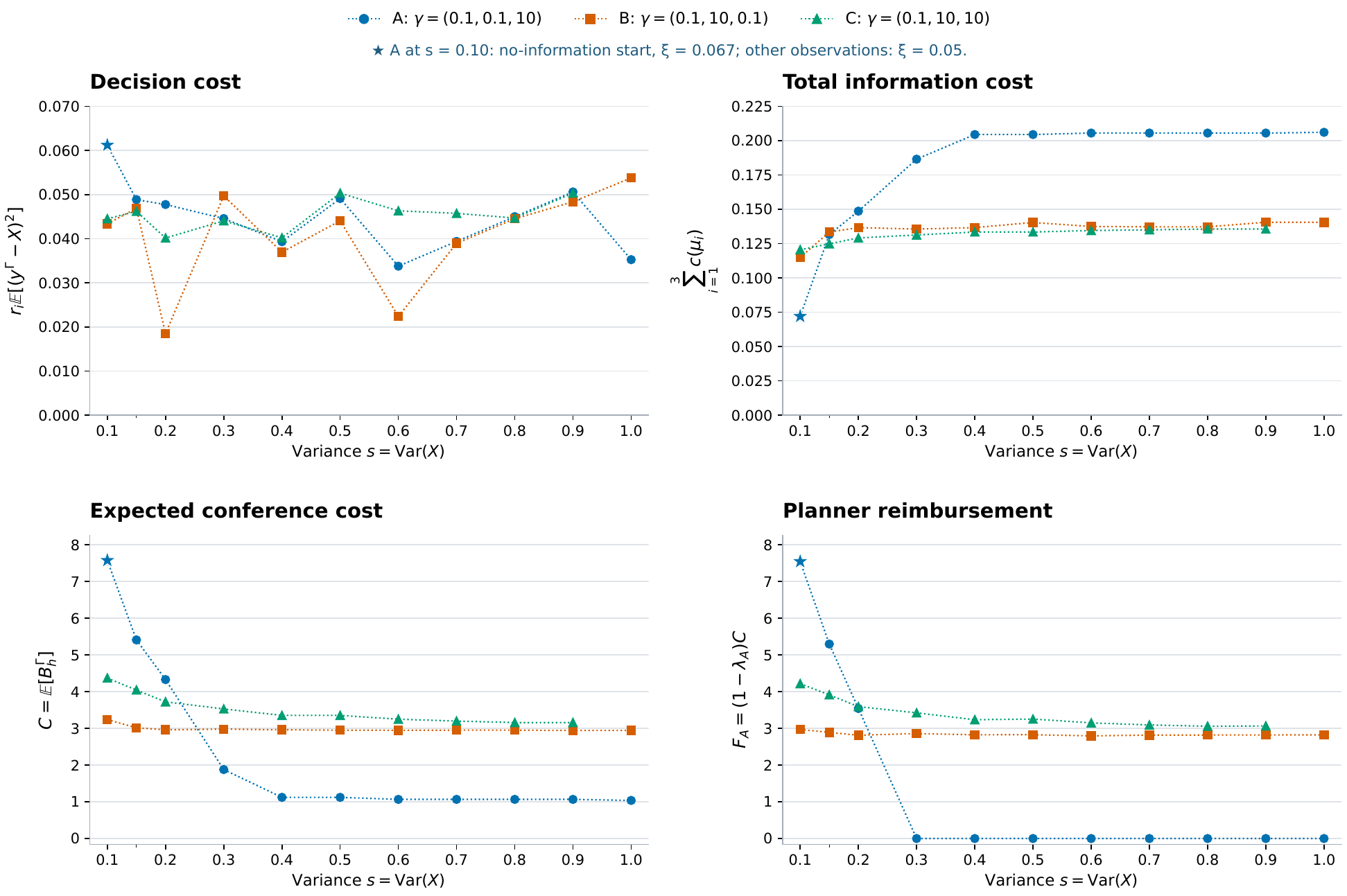} 
\caption{Decision outcomes and reimbursement across variance} 
\label{fig:combined}
\end{center}
\end{figure} 
The top left panel shows that having two demanding nonmediators does not uniformly improve decision accuracy. 
Case C produces less accurate decisions than both A and B at $s=0.40,0.50,0.60$ and $0.70$, although it is more accurate than both at $s=0.15$ and $0.30$.

The top right panel shows that the information cost substantially rises in A, while B and C remain relatively stable. 
The bottom left panel shows that A's conference cost falls sharply. 
At every plotted variance $s \geq 0.30$, it is lower than B and C.  

The player-level decomposition in Figure \ref{fig:by-player} shows that the mediator's information cost in A reached approximately 
$0.068$ at $s=0.30$ and remains at that level, whereas it is small in B and zero in C. Thus, A's conference-cost advantage is associated with 
substantially greater mediator information acquisition.

The bottom right panel shows that 
at every tested variance $s \geq 0.30$, A's certified reimbursement is zero. 
This lowers the planner's fiscal expenditure while assigning the conference cost to the mediator. 
A planner that also values reducing the mediator's private burden might therefore prefer B, despite its higher conference cost and higher certified sufficient 
reimbursement. Such a preference depends on how the planner weights decision accuracy, information acquisition, fiscal expenditure, and the distribution of costs. 

\begin{figure}[H]
\begin{center}
\includegraphics[width=16cm]{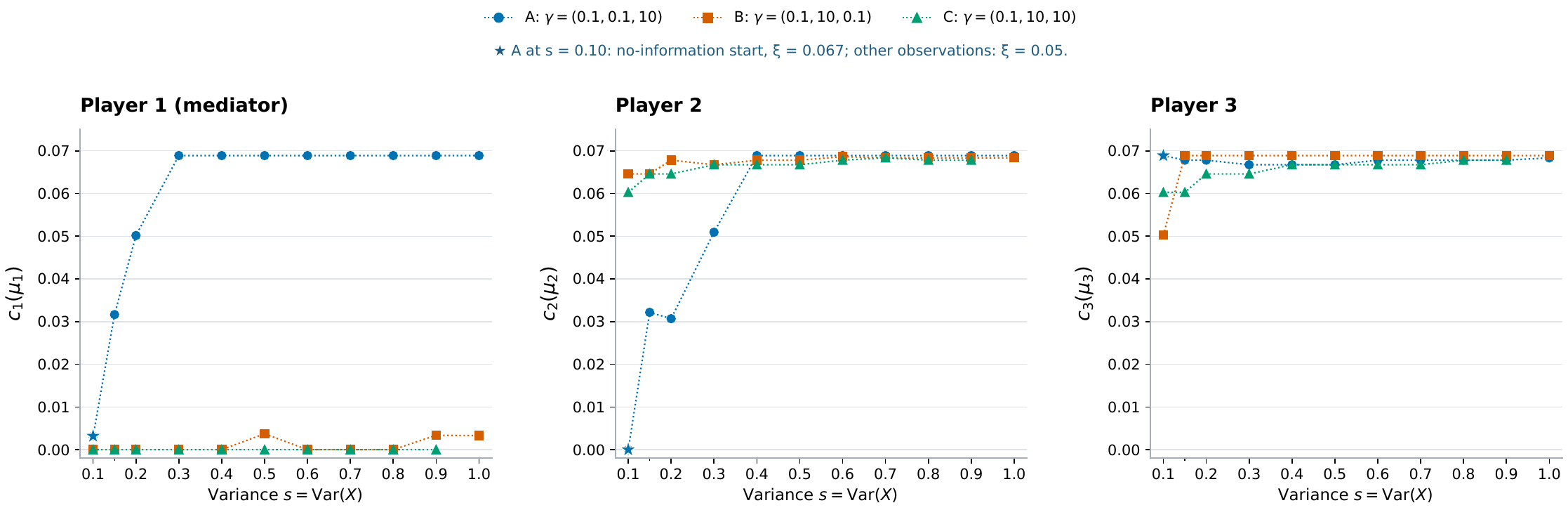} \\ 
\caption{Information costs by players} 
\label{fig:by-player}
\end{center}
\end{figure} 

\section{Future Work}\label{sec:future}
The study jointly examines three elements: (i) DeGroot-type learning with an acceptance rule based on KL divergence, (ii) conference structure supporting a global consensus, and (iii) ex ante private information acquisition that endogenously determines players' acceptance tolerances. 
In our environment, we establish sufficient conditions for consensus and 
prove the existence of a pure-strategy Nash equilibrium in the information-acquisition game under full reimbursement. 
 The  numerical examples suggest that placing the demanding player last in 
the bilateral conference cycle is associated with lower certified sufficient planner reimbursement at larger state variances, but 
greater information acquisition burden for the mediator. 

The computational analysis has three main limitations. First, the study does not establish a general upper bound on the total number of 
computational iterations to estimate computation time. Second, the computed $\xi$-Nash equilibrium selected by the procedure may depend on the initial profile and update rule, so the numerical comparisons need not hold across all equilibria. Third, the numerical analysis covers only three players and a limited range of conference structures. Future research could overcome these limitations.

\bibliography{jecon_ref3}
\begin{appendices}
\counterwithin{theorem}{section}
\counterwithin{lemma}{section}
\counterwithin{proposition}{section}
\numberwithin{equation}{section}

\section{Proof of Proposition \ref{prop:naive-listener}}

Fix a naive listener $M=i^*$ and write
\[
x_i(t):=\bigl(p_i(t)(E_1),\ldots,p_i(t)(E_\tau)\bigr),
\qquad
K:=\operatorname{co}\{x_i(0):i\in[n]\}.
\]
For probability vectors, write
$\|x-y\|_{\mathrm{TV}}:=\frac12\|x-y\|_1$.
Every update is a convex combination, so $x_i(t)\in K$ for all $i,t$. Initial strict positivity makes $K$ compact and bounded away from the boundary of the partition simplex.
Since $D_{\mathrm{KL}}$ is uniformly continuous on $K\times K$ and $D_{\mathrm{KL}}(x\|x)=0$, choose 
$\varepsilon_i\in(0,1]$ such that
\[
x,y\in K,\quad
\|x-y\|_{\mathrm{TV}}<\varepsilon_i
\quad\Longrightarrow\quad
D_{\mathrm{KL}}(x\|y)<\eta_i.
\]
Set
\[
\varepsilon_0:=\min_{i\in[n]}\varepsilon_i>0,
\qquad
\beta:=\varepsilon_0/4.
\]

\begin{lemma}\label{lem:alpha-range}
Suppose a bilateral stage between $M$ and $j\ne M$ starts at time $s$, with $u=x_M(s)$ and $v=x_j(s)$.
Repeating this conference until $j$ accepts $M$ takes finitely many updates and gives both players the partition belief
\[
\alpha u+(1-\alpha)v,
\qquad \alpha\in[\beta,1/2].
\]
\end{lemma}

\begin{proof}
Let $\tilde t\ge s$ be the first time at which $j$ accepts $M$. Before this acceptance,
\[
x_j(t)=v,\qquad
x_M(t)=v+2^{-(t-s)}(u-v).
\]
If acceptance never occurred, the latter expression would converge to $v$, contradicting $\eta_j>0$ and continuity of KL. Thus $\tilde t<\infty$.
At the accepted update, both players adopt
\[
x_M(\tilde t+1)=x_j(\tilde t+1)
=v+\alpha(u-v),
\qquad
\alpha:=2^{-(\tilde t-s+1)}\le\tfrac12.
\]

If $\tilde t=s$, then $\alpha=1/2\ge\beta$.
Otherwise, rejection at $\tilde t-1$ and the defining property of $\varepsilon_j$ imply
\[
\|x_M(\tilde t-1)-v\|_{\mathrm{TV}}
\ge\varepsilon_j.
\]
Consequently,
\[
\alpha\|u-v\|_{\mathrm{TV}}
=\tfrac14\|x_M(\tilde t-1)-v\|_{\mathrm{TV}}
\ge\tfrac14\varepsilon_j.
\]
Since $\|u-v\|_{\mathrm{TV}}\le1$, this gives
$\alpha\ge\varepsilon_j/4\ge\beta$.
\end{proof}

\begin{proof}[Proof of Proposition \ref{prop:naive-listener}]
The case $n=1$ is immediate. Suppose $n\ge2$ and fix an order $j_1,\ldots,j_N$ of $[n]\setminus\{M\}$, where $N=n-1$. In each cycle, repeat the conference $\{M,j_r\}$ until $j_r$ accepts $M$, for $r=1,\ldots,N$.
By Lemma~\ref{lem:alpha-range}, every stage is finite and its completed update has coefficient $\alpha_r\in[\beta,1/2]$.

Write $x_i$ and $x_i^+$ for the partition beliefs before and after one complete cycle. Its update is
\[
x_i^+=\sum_{\ell\in[n]}A_{i\ell}x_\ell,
\]
where $A$ is row-stochastic. The coefficient of the initial listener belief is
\[
A_{j_rM}=\prod_{k=1}^{r}\alpha_k,
\qquad
A_{MM}=\prod_{k=1}^{N}\alpha_k.
\]
Hence $A_{iM}\ge\gamma$ for every $i\in[n]$, where
\[
\gamma:=\prod_{r=1}^{N}\alpha_r\ge\beta^N.
\]
Thus, for every $i\in[n]$, there exists $y_i\in\operatorname{co}\{x_\ell:\ell\in[n]\}$ such that
\[
x_i^+=\gamma x_M+(1-\gamma)y_i.
\]
It follows that
\[
\max_{i,j\in[n]}\|x_i^+-x_j^+\|_{\mathrm{TV}}
\le
(1-\beta^N)
\max_{i,j\in[n]}\|x_i-x_j\|_{\mathrm{TV}}.
\]

Let $T_q$ be the starting time of cycle $q$, with $T_0=0$, and define
\[
D_q:=\max_{i,j\in[n]}
\|x_i(T_q)-x_j(T_q)\|_{\mathrm{TV}}.
\]
Every $T_q$ is finite, and the contraction gives
\[
D_q\le c^qD_0,
\qquad c:=1-\beta^N\in(0,1).
\]
Choose a finite $q_*$ such that $D_{q_*}<\varepsilon_0$ and set $T:=T_{q_*}$. Then
\[
D_{\mathrm{KL},\mathcal E}(p_j(T)\|p_i(T))
=
D_{\mathrm{KL}}(x_j(T)\|x_i(T))
<\eta_i
\qquad\forall i,j\in[n].
\]
Therefore, everyone accepts everyone else at time $T$.
The grand conference $C(T)=[n]$ gives
\[
p_i(T+1)=\frac1n\sum_{j=1}^{n}p_j(T)
\qquad\forall i\in[n],
\]
which is exact global consensus.
\end{proof}

\section{Proof of Lemma \ref{lem:mediated-consensus-continuity}}
\begin{lemma}\label{lem:mediated-pair-continuity}
For $\eta_i>0$, the coefficient $\theta_i(u,v;\eta_i)$ and the completed agreement $F_i(u,v;\eta_i)$ are jointly continuous in $(u,v,\eta_i)$. In the rejection case, the boundary opinion $a_i:=v+\theta_i(u,v;\eta_i)(u-v)$ is uniquely determined by the segment restriction.
\end{lemma}
\begin{proof}
If $u\ne v$, let
$g(\theta)=D_{\mathrm{KL}}(v+\theta(u-v)\|v)$.
All coordinates are bounded below by $\varepsilon$, and
\[
g(0)=g'(0)=0,\qquad
g''(\theta)
=\sum_{k=1}^{\tau}
\frac{(u(k)-v(k))^2}{v(k)+\theta(u(k)-v(k))}>0.
\]
Hence $g$ is continuous and strictly increasing on $[0,1]$. If $g(1)>\eta_i>0$, it has a unique root $g(\theta_i)=\eta_i$ in $(0,1)$. Joint continuity of $g$ and strict monotonicity give continuity of this root. At $g(1)=\eta_i$, the rejection-side root tends to $1$, which matches the immediate-acceptance branch. If $u=v$, then $g(1)=0<\eta_i$, so $\theta_i=1$ in a neighborhood of that point. Formula 
$F_i(u,v;\eta_i)
:=
v+\frac{\theta_i(u,v;\eta_i)}{2}(u-v)$ 
 then gives continuity of $F_i$.
\end{proof}

\begin{proof}[Proof of Lemma \ref{lem:mediated-consensus-continuity}]
Fix $\eta^0\in(0,\infty)^n$ and a compact neighborhood $V\subset(0,\infty)^n$ of $\eta^0$. By Lemma~\ref{lem:mediated-pair-continuity}, each $\theta_i$ is continuous and strictly positive. Compactness therefore gives constants $\beta>0$ and $\underline\eta>0$ such that
\[
\frac{\theta_i(u,v;\eta_i)}2\in[\beta,1/2],
\qquad
\min_i\eta_i\geq\underline\eta,
\]
uniformly over $u,v\in K_\varepsilon$, $\eta\in V$, and targets $i\ne h$.

Continue the bilateral dynamics without applying the stopping rule. Let $m\in\mathbb N_0$ denote the number of completed
bilateral cycles. The cycle-contraction argument in the proof of Proposition~\ref{prop:naive-listener} gives
\[
\max_{a,b\in[n]}
\bigl\|q_a^{m(n-1)}-q_b^{m(n-1)}\bigr\|_1
\leq 2c^m,
\qquad c:=1-\beta^{n-1}<1.
\]
Since $D_{\mathrm{KL}}(u\|v)\leq\varepsilon^{-1}\|u-v\|_1^2$, choose $M$ such that $4c^{2M}/\varepsilon\leq\underline\eta$.
Then
$\ell^\Gamma(\pi;\eta)\leq r':=M(n-1)$
uniformly over $K_\varepsilon^n\times V$.

After the stopping condition first holds, continued bilateral updates remain in the convex hull of the opinion profile at that 
stopping time. Joint convexity of KL therefore preserves the stopping condition. These updates are consequently midpoint updates and preserve the global average. Hence, even if the actual protocol stops earlier,
\[
p^{*,\Gamma}(\pi;\eta)
=\frac1n\sum_{j=1}^n q_j^{r'}(\pi;\eta).
\]
By Lemma~\ref{lem:mediated-pair-continuity}, this is a fixed finite composition of continuous bilateral updates
followed by averaging. Thus $p^{*,\Gamma}$ is jointly continuous, as is $y^\Gamma=\sum_k x_k p^{*,\Gamma}(k)$.
Since $\eta^0$ was arbitrary, the conclusion follows.
\end{proof}

\section{Proof of Lemma \ref{lem:stopped-cost-regularity}}

\begin{lemma}[Weak$^*$ compactness of admissible experiments]
\label{lem:admissible-experiment}
The probability-measure space $\Delta(K_\varepsilon)$ and the admissible experiment set
\[
\mathcal D_i
:=\left\{\mu_i\in\Delta(K_\varepsilon):
\int_{K_\varepsilon}\pi\,\mu_i(d\pi)=p_{\mathcal E}\right\}
\]
are weak$^\ast$ compact in $C(K_\varepsilon)^*$.
Moreover, $\mathcal D_i$ is nonempty and convex.
Consequently, $\mathcal D:=\prod_{i=1}^n\mathcal D_i$ is nonempty, compact, and convex in the product weak$^\ast$ topology.
\end{lemma}

\begin{proof}
Let $E=C(K_\varepsilon)$, equipped with the supremum norm, and identify $E^*$ with finite signed Borel measures
by the Riesz representation theorem. By Banach--Alaoglu, the closed unit ball
\[
\mathbb B_{E^*}
:=\{\zeta\in E^*:\|\zeta\|_{E^*}\le1\}
\]
is weak$^\ast$ compact.

Every probability measure has dual norm one, and
\[
\Delta(K_\varepsilon)
=
\left\{\zeta\in\mathbb B_{E^*}:
\langle1,\zeta\rangle=1,\quad
\langle g,\zeta\rangle\ge0
\text{ for every }g\in E,\ g\ge0
\right\}.
\]
For each fixed $g\in E$, the map $\zeta\mapsto\langle g,\zeta\rangle$ is weak$^\ast$ continuous.
Thus the mass and positivity conditions define a weak$^\ast$ closed subset of $\mathbb B_{E^*}$.
Hence $\Delta(K_\varepsilon)$ is weak$^\ast$ compact.

For each $k$, the coordinate function $g_k(\pi):=\pi(k)$ is continuous, so
\[
\mathcal D_i
=
\Delta(K_\varepsilon)\cap
\bigcap_{k=1}^{\tau}
\{\zeta\in E^*:
\langle g_k,\zeta\rangle=p_{\mathcal E}(k)\}
\]
is weak$^\ast$ closed and therefore compact.
Since $p_{\mathcal E}\in K_\varepsilon$, $\delta_{p_{\mathcal E}}\in\mathcal D_i$, proving nonemptiness. Convexity follows from linearity of the posterior-mean constraint.
The finite product $\mathcal D$ is therefore nonempty, compact, and convex.
\end{proof}

\begin{proof}[Proof of Lemma \ref{lem:stopped-cost-regularity}]
Consider the bilateral continuation without applying the stopping rule, so that $q^s$ and $\theta_r^\Gamma$ are defined at every stage. Write
\[
d(q):=\max_{a,b\in[n]}D_{\mathrm{KL}}(q_a\|q_b),
\qquad
\eta_{\min}:=\min_{i\in[n]}\eta_i.
\]
Every finite-stage profile $q^s$ and coefficient $\theta_r^\Gamma$ is continuous in $(\pi,\eta)$ by
Lemma~\ref{lem:mediated-pair-continuity}.
Since $d$ and $\eta_{\min}$ are continuous, for $r\ge1$ the set
\[
\{\ell^\Gamma\ge r\}
=
\bigcap_{s=0}^{r-1}
\{d(q^s)>\eta_{\min}\}
\]
is open. Its indicator is therefore lower semicontinuous.
Because $2-\theta_r^\Gamma$ is continuous and nonnegative, each term in
\[
B_h^\Gamma
=
1+\sum_{r=1}^{\infty}
\mathbf1_{\{\ell^\Gamma\ge r\}}
(2-\theta_r^\Gamma)
\]
is lower semicontinuous. The partial sums increase to $B_h^\Gamma$, so $B_h^\Gamma$ is jointly lower
semicontinuous. Moreover, $1\le2-\theta_r^\Gamma<2$ gives
\[
1+\ell^\Gamma
\le B_h^\Gamma
\le1+2\ell^\Gamma.
\]

The proof of Lemma~\ref{lem:mediated-consensus-continuity} gives a finite stage bound locally uniformly in the positive tolerance vector and uniformly over initial profiles. The preceding inequality therefore gives local boundedness of $B_h^\Gamma$.

It remains to obtain a common bound across all admissible experiments. Since $\pi\mapsto\operatorname{Var}_\pi(X)$ is continuous,
\[
R_i(\mu_i)
=\int_{K_\varepsilon}
\operatorname{Var}_\pi(X)\,\mu_i(d\pi)
\]
is weak$^\ast$ continuous.
By Lemma~\ref{lem:admissible-experiment}, its image $\mathcal Z_i(p_{\mathcal E})=R_i(\mathcal D_i)$ is nonempty and compact.
Assumption~\ref{ass:positive-endogenous-tolerance} therefore implies
\[
\begin{aligned}
\underline\eta
&:=\min_{i\in[n]}
\min_{z\in\mathcal Z_i(p_{\mathcal E})}
H(z;\gamma_i)>0,\\
\overline\eta
&:=\max_{i\in[n]}
\max_{z\in\mathcal Z_i(p_{\mathcal E})}
H(z;\gamma_i)<\infty.
\end{aligned}
\]
Thus every attainable tolerance vector lies in the compact set
\[
\mathcal T:=[\underline\eta,\overline\eta]^n
\subset(0,\infty)^n.
\]
Cover $\mathcal T$ by finitely many neighborhoods on which the proof of Lemma~\ref{lem:mediated-consensus-continuity}
provides uniform stage bounds $M_1,\ldots,M_L$.
Then $M:=\max_{\ell=1,\ldots,L}M_\ell$ bounds the number of bilateral stages for every admissible experiment profile and every initial posterior profile.
Consequently,
\[
1\le
B_h^\Gamma(\pi;\eta[\mu;\gamma])
\le1+2M.
\]
\end{proof}

\section{Lemma \ref{lem:acquisition-payoff-regularity}}
\begin{proof}
The weak$^\ast$ topology on $\Delta(K_\varepsilon)$ is metrizable because $K_\varepsilon$ is compact metric \citep[][Thm 15.11]{aliprantis2006infinite}. It suffices to take an arbitrary sequence satisfying
\begin{equation}\label{eq:acquisition-weakstar-sequence}
\mu^m,\mu\in\mathcal D,\qquad
\int g\,d\mu_j^m\longrightarrow\int g\,d\mu_j
\quad(g\in C(K_\varepsilon),\ j\in[n]).
\end{equation}

\medskip
\noindent\textit{(i) Conditional posterior laws.}
For $k\in\Theta$, define
\[
\rho_{j,\mu}^k(dq):=\frac{q(k)}{p_k}\mu_j(dq),\qquad
\rho_\mu^k:=\bigotimes_{j=1}^n\rho_{j,\mu}^k,
\qquad
\nu_\mu=\sum_{k=1}^{\tau}p_k\rho_\mu^k.
\]
Bayes plausibility gives
\[
\rho_{j,\mu}^k(K_\varepsilon)
=\frac1{p_k}\int_{K_\varepsilon}q(k)\,\mu_j(dq)=1.
\]
For $g\in C(K_\varepsilon)$,
\[
q\longmapsto g(q)\frac{q(k)}{p_k}\in C(K_\varepsilon),
\]
and hence
\begin{align*}
\int g\,d\rho_{j,\mu^m}^k
&=\int g(q)\frac{q(k)}{p_k}\,\mu_j^m(dq)\\
&\longrightarrow
\int g(q)\frac{q(k)}{p_k}\,\mu_j(dq)
=\int g\,d\rho_{j,\mu}^k.
\end{align*}

For the product measures, let
\[
\mathscr A
:=\left\{
\sum_{a=1}^N\prod_{j=1}^n g_{ja} \circ pr_j:
N\in\mathbb N,\ g_{ja}\in C(K_\varepsilon)
\right\}, \quad pr_j(\pi)=\pi_j \quad \pi \in (K_\varepsilon)^n.
\]
The algebra $\mathscr A$ contains constants and separates points. By Stone--Weierstrass,
\[
\overline{\mathscr A}^{\|\cdot\|_\infty}
=C(K_\varepsilon^n).
\]
For $F(\pi)=\sum_{a=1}^N\prod_jg_{ja}(\pi_j)\in\mathscr A$,
\begin{align*}
\int F\,d\rho_{\mu^m}^k
&=\sum_{a=1}^N\prod_{j=1}^n
\int g_{ja}\,d\rho_{j,\mu^m}^k\\
&\longrightarrow
\sum_{a=1}^N\prod_{j=1}^n
\int g_{ja}\,d\rho_{j,\mu}^k
=\int F\,d\rho_\mu^k.
\end{align*}
Thus, for $g\in C(K_\varepsilon^n)$ and $F\in\mathscr A$,
\[
\limsup_m\left|
\int g\,d\rho_{\mu^m}^k-\int g\,d\rho_\mu^k
\right|
\leq2\|g-F\|_\infty.
\]
Taking the infimum over $F\in\mathscr A$: 
\[
\inf_{F \in \mathscr A} \|g-F\|_\infty = 0
\]
gives
\[
\rho_{\mu^m}^k\overset{*}{\rightharpoonup}\rho_\mu^k
\quad(k\in\Theta),\qquad
\nu_{\mu^m}=\sum_kp_k\rho_{\mu^m}^k
\overset{*}{\rightharpoonup}\sum_kp_k\rho_\mu^k=\nu_\mu.
\]

\medskip
\noindent\textit{(ii) Endogenous tolerances and augmented laws.}
Set
\[
v(q):=\Var_q(X)
=\sum_kq(k)x_k^2-\left(\sum_kq(k)x_k\right)^2
\in C(K_\varepsilon).
\]
Then
\begin{align*}
R_j(\mu_j^m)
&=\int v\,d\mu_j^m
\longrightarrow\int v\,d\mu_j=R_j(\mu_j),\\
\eta_j^m:=H(R_j(\mu_j^m);\gamma_j)
&\longrightarrow H(R_j(\mu_j);\gamma_j)=:\eta_j>0.
\end{align*}
Choose
\[
T:=\prod_{j=1}^n[\eta_j/2,3\eta_j/2]
\subset(0,\infty)^n.
\]
There is $m_0$ such that $\eta^m,\eta\in T$ for $m\geq m_0$.
For $\psi\in C(T)$,
\[
\int_T\psi\,d\delta_{\eta^m}
=\psi(\eta^m)\longrightarrow\psi(\eta)
=\int_T\psi\,d\delta_\eta,
\qquad
\delta_{\eta^m}\overset{*}{\rightharpoonup}\delta_\eta.
\]
Write
\[
\Lambda_m^k:=\rho_{\mu^m}^k\otimes\delta_{\eta^m},
\qquad
\Lambda^k:=\rho_\mu^k\otimes\delta_\eta.
\]
For $g\in C(K_\varepsilon^n\times T)$, uniform continuity on the compact domain gives
\[
\sup_{\pi\in K_\varepsilon^n}
|g(\pi,\eta^m)-g(\pi,\eta)|\longrightarrow0.
\]
Consequently,
\begin{align*}
\left|\int g\,d\Lambda_m^k-\int g\,d\Lambda^k\right|
&\leq\sup_{\pi\in K_\varepsilon^n}
|g(\pi,\eta^m)-g(\pi,\eta)|\\
&\quad+
\left|\int g(\pi,\eta)\,\rho_{\mu^m}^k(d\pi)
-\int g(\pi,\eta)\,\rho_\mu^k(d\pi)\right|\\
&\longrightarrow0.
\end{align*}
Therefore, on $K_\varepsilon^n\times T$,
\begin{align*}
\Lambda_m^k&\overset{*}{\rightharpoonup}\Lambda^k,\\
\Lambda_m:=\nu_{\mu^m}\otimes\delta_{\eta^m}
&=\sum_kp_k\Lambda_m^k
\overset{*}{\rightharpoonup}\sum_kp_k\Lambda^k
=\nu_\mu\otimes\delta_\eta=:\Lambda.
\end{align*}

\medskip
\noindent\textit{(iii) Continuous decision loss and acquisition cost.}
For every experiment profile, define
\begin{align}
L(\mu)
&:=\sum_{k=1}^{\tau}p_k
\int_{K_\varepsilon^n}
\bigl(y^\Gamma(\pi;\eta[\mu;\gamma])-x_k\bigr)^2
\,\rho_\mu^k(d\pi),\label{eq:acquisition-ex-ante-loss}\\
C(\mu)
&:=\int_{K_\varepsilon^n}
B_h^\Gamma(\pi;\eta[\mu;\gamma])\,\nu_\mu(d\pi).
\label{eq:acquisition-ex-ante-cost}
\end{align}
The law of $(J,\Pi)$ is
$\Pr_\mu\{J=k,\Pi\in d\pi\}=p_k\rho_\mu^k(d\pi)$; thus
$L(\mu)=\E_\mu[(y^\Gamma-X)^2]$ under this joint law.
By Lemma~\ref{lem:mediated-consensus-continuity},
\[
\ell_k(\pi,t):=\bigl(y^\Gamma(\pi;t)-x_k\bigr)^2
\in C(K_\varepsilon^n\times T).
\]
With $\Delta_x:=\max_kx_k-\min_kx_k$,
\[
0\leq\ell_k(\pi,t)\leq\Delta_x^2,
\qquad
L(\mu^m)=\sum_kp_k\int\ell_k\,d\Lambda_m^k
\longrightarrow\sum_kp_k\int\ell_k\,d\Lambda^k=L(\mu).
\]
Also,
\[
d_p(q):=D_{\mathrm{KL}}(q\|p_{\mathcal E})
=\sum_kq(k)\log\frac{q(k)}{p_k}
\in C(K_\varepsilon),
\]
so
\[
c_i(\mu_i^m)=\kappa_i\int d_p\,d\mu_i^m
\longrightarrow\kappa_i\int d_p\,d\mu_i=c_i(\mu_i).
\]

\medskip
\noindent\textit{(iv) Portmanteau inequality for conference cost.}
By Lemma~\ref{lem:stopped-cost-regularity}, for some finite $M_T$,
\[
B_h^\Gamma\text{ is lower semicontinuous on }K_\varepsilon^n\times T,
\qquad 1\leq B_h^\Gamma\leq1+2M_T.
\]
Theorem~15.5 of \cite{aliprantis2006infinite} gives
\begin{align*}
C(\mu)
&=\int B_h^\Gamma\,d\Lambda\\
&\leq\liminf_{m\to\infty}\int B_h^\Gamma\,d\Lambda_m
=\liminf_{m\to\infty}C(\mu^m).
\end{align*}

\medskip
\noindent\textit{(v) Regularity of the reimbursed utilities.}
The ex ante utilities satisfy
\begin{equation}\label{eq:acquisition-payoff-decomposition}
U_i^\lambda(\mu;\Gamma)
=-r_iL(\mu)-c_i(\mu_i)-\lambda \mathbf{1}_{\{i=h\}}C(\mu).
\end{equation}
Thus continuity follows when $\lambda \mathbf{1}_{\{i=h\}}=0$.
For the mediator with $\lambda>0$,
\begin{align*}
\limsup_{m\to\infty}U_h^\lambda(\mu^m;\Gamma)
&=-r_hL(\mu)-c_h(\mu_h)-\lambda\liminf_{m\to\infty}C(\mu^m)\\
&\leq-r_hL(\mu)-c_h(\mu_h)-\lambda C(\mu)
=U_h^\lambda(\mu;\Gamma).
\end{align*}
For every $\rho\in K_\varepsilon$, the version specified in
\eqref{eq:posterior-conditional-law} is
\[
\Pr_\mu(J=k,\Pi_{-i}\in d\pi_{-i}\mid\Pi_i=\rho)
=\rho(k)\bigotimes_{j\ne i}\rho_{j,\mu}^k(d\pi_j).
\]
The factors are probability measures by Bayes plausibility, and integrating this kernel against $\mu_i(d\rho)$ recovers the joint law.
Thus it specifies a conditional law even outside
$\operatorname{supp}(\mu_i)$.
For fixed opponents' experiments and tolerances, the product measures are fixed and the weights $\rho(k)$ are continuous.
Continuity and boundedness of the squared decision loss give continuity of $V_i^\Gamma$ in $\rho$ by dominated convergence.
Lower semicontinuity and nonnegativity of $B_h^\Gamma$ give lower semicontinuity of $C_i$ by Fatou's lemma.
The assertions for $f_i^\lambda$ follow.
\end{proof}
\section{Proof of Proposition \ref{prop:acquisition-equilibrium-existence}}
\begin{proof}
By Lemma~\ref{lem:admissible-experiment}, $\mathcal D=\prod_{i=1}^n\mathcal D_i$ is nonempty
and compact in the product weak$^\ast$ topology.
The proof of Lemma~\ref{lem:acquisition-payoff-regularity} establishes continuity of $L$ and each $c_i$, and hence of $\Phi$.
Therefore $\Phi$ attains its maximum on $\mathcal D$.
Let $\mu^\ast$ be any global maximizer.
For every player $i$ and every admissible deviation $\widehat\mu_i\in\mathcal D_i$,
\[
\begin{aligned}
&U_i^0(\widehat\mu_i,\mu_{-i}^\ast;\Gamma)
-U_i^0(\mu^\ast;\Gamma)\\
&\qquad=r_i\left[
\Phi(\widehat\mu_i,\mu_{-i}^\ast)-\Phi(\mu^\ast)\right]\leq0.
\end{aligned}
\]
Hence $\mu_i^\ast\in\operatorname{BR}_i^\Gamma(\mu_{-i}^\ast;\gamma,0)$ for all $i$.
This proves existence without assuming quasiconcavity of payoffs in the experiment choice.
\end{proof}

\section{Proof of Proposition \ref{prop:existence-costly-information-case}}
\begin{proof}
First we show
\[
L(\mu_i,\mu_{-i}^0)\geq R_i(\mu_i).
\]
Put \(p=p_{\mathcal E}\), and define
\[
m(\pi):=\sum_k\pi(k)x_k,
\qquad
m_0:=\sum_kp(k)x_k,
\qquad
V_0:=\sum_kp(k)x_k^2-m_0^2.
\]
By Bayes plausibility,
\[
\int_{K_\varepsilon}\pi(k)\,\mu_i(d\pi)=p(k)
\qquad\text{for every }k.
\]
Consequently,
\begin{align*}
\int_{K_\varepsilon}m(\pi)\,\mu_i(d\pi)
&=m_0.
\end{align*}
Since \(\mu_i\) is a probability measure, it follows that
\begin{align*}
\int_{K_\varepsilon}
\bigl(m(\pi)-m_0\bigr)^2\,\mu_i(d\pi)
&=
\int_{K_\varepsilon}m(\pi)^2\,\mu_i(d\pi)
-2m_0\int_{K_\varepsilon}m(\pi)\,\mu_i(d\pi)
+m_0^2\\
&=
\int_{K_\varepsilon}m(\pi)^2\,\mu_i(d\pi)
-m_0^2.
\end{align*}
Therefore,
\begin{align*}
R_i(\mu_i)
&=
\int_{K_\varepsilon}
\operatorname{Var}_\pi(X)\,\mu_i(d\pi)\\
&=
\int_{K_\varepsilon}
\left[
\sum_k\pi(k)x_k^2-m(\pi)^2
\right]\mu_i(d\pi)\\
&=
\sum_kx_k^2
\int_{K_\varepsilon}\pi(k)\,\mu_i(d\pi)
-
\int_{K_\varepsilon}m(\pi)^2\,\mu_i(d\pi)\\
&=
\sum_kp(k)x_k^2
-
\int_{K_\varepsilon}m(\pi)^2\,\mu_i(d\pi)\\
&=
\sum_kp(k)x_k^2-m_0^2
-
\left[
\int_{K_\varepsilon}m(\pi)^2\,\mu_i(d\pi)-m_0^2
\right]\\
&=
V_0
-
\int_{K_\varepsilon}
\bigl(m(\pi)-m_0\bigr)^2\,\mu_i(d\pi)\\
&\leq V_0.
\end{align*}
Since all opponents remain uninformative, conditional on
\(\Pi_i=\pi\), 
\begin{align*}
\mathbb E_\mu&[y^\Gamma(\pi,(p_\mathcal{E})_{-i};\eta)-X)^2\mid\Pi_i=\pi] \\
&=
\operatorname{Var}_{\pi}(X)
+\bigl(y^\Gamma(\pi,(p_\mathcal{E})_{-i};\eta)-\mathbb E_\pi[X]\bigr)^2
\geq \operatorname{Var}_{\pi}(X).
\end{align*}
Integrating with respect to \(\mu_i\) gives
\(L(\mu_i,\mu_{-i}^0)\geq R_i(\mu_i)\).

The law of total variance gives
\[
V_0-R_i(\mu_i)
=\int_{K_\varepsilon}
\bigl(\mathbb E_\pi[X]-m_0\bigr)^2\,\mu_i(d\pi).
\]
Using total variation and Pinsker’s inequality, with natural logarithms,
\[
\bigl(\mathbb E_\pi[X]-m_0\bigr)^2
\leq \Delta_x^2\|\pi-p\|_{\mathrm{TV}}^2
\leq \frac{\Delta_x^2}{2}D_{\mathrm{KL}}(\pi\|p).
\]
The first inequality follows from: 
Write
\[
x_{\min}:=\min_kx_k,
\qquad
x_{\max}:=\max_kx_k,
\qquad
c:=\frac{x_{\max}+x_{\min}}{2}.
\]
Then
\[
|x_k-c|\leq\frac{\Delta_x}{2}
\qquad\text{for every }k.
\]
Using the convention
\[
\|\pi-p\|_{\mathrm{TV}}
:=\frac12\sum_k|\pi(k)-p(k)|,
\]
we obtain
\begin{align*}
\bigl(\mathbb E_\pi[X]-m_0\bigr)^2
&=
\left|
\sum_kx_k\bigl(\pi(k)-p(k)\bigr)
\right|^2\\
&=
\left|
\sum_k(x_k-c)\bigl(\pi(k)-p(k)\bigr)
\right|^2\\
&\leq
\left(
\sum_k|x_k-c|\,|\pi(k)-p(k)|
\right)^2\\
&\leq
\left(
\frac{\Delta_x}{2}
\sum_k|\pi(k)-p(k)|
\right)^2\\
&=
\Delta_x^2\|\pi-p\|_{\mathrm{TV}}^2.
\end{align*}

Set
\[
I_i(\mu_i):=\int_{K_\varepsilon}
D_{\mathrm{KL}}(\pi\|p)\,\mu_i(d\pi).
\]
Since \(C(\mu_i,\mu_{-i}^0)\geq1\), the deviation gain satisfies
\[
\begin{aligned}
&U_i^\lambda(\mu_i,\mu_{-i}^0;\Gamma)
-U_i^\lambda(\mu^0;\Gamma)\\
&\quad\leq
\left(\frac{r_i\Delta_x^2}{2}-\kappa_i\right)I_i(\mu_i)
-\lambda\mathbf1_{\{i=h\}}
\bigl[C(\mu_i,\mu_{-i}^0)-1\bigr]
\leq0.
\end{aligned}
\]
This holds for every player, including the mediator, and every admissible deviation. Thus \(\mu^0\) is a Nash equilibrium in $\mathcal G^\lambda(\Gamma,\gamma)$. 
\end{proof}

\section{Proof of Proposition \ref{prop:strong-duality}}
\begin{proof}
Write $K=K_\varepsilon$, $p=p_{\mathcal E}$, and
$v(\pi)=\operatorname{Var}_\pi(X)$.
Summing the posterior-mean constraints gives
\[
\mu_i(K)
=\int_K\sum_{k=1}^{\tau}\pi(k)\,\mu_i(d\pi)
=\sum_{k=1}^{\tau}p_k=1.
\]
Thus the feasible set of $(P_i)$ is
\[
\mathcal F_i(z_i)
:=\left\{\mu_i\in\Delta(K):
\int_K\pi\,d\mu_i=p,\quad
\int_Kv\,d\mu_i=z_i\right\}.
\]
It is nonempty because $z_i$ is attainable.

Since $K$ is compact metric, $\Delta(K)$ is compact and metrizable in the weak$^\ast$ topology.
The moment functions $\pi(k)$ and $v$ are continuous, so their integrals are weak$^\ast$ continuous.
Hence $\mathcal F_i(z_i)$ is a closed subset of $\Delta(K)$ and is therefore compact and metrizable.
Compactness is equivalent to sequential compactness in metric spaces \citep[Theorem~3.28]{aliprantis2006infinite}.
Consequently, every sequence in $\mathcal F_i(z_i)$ has a weak$^\ast$ convergent subsequence with its limit
in $\mathcal F_i(z_i)$.

Let $s:=\operatorname{val}(P_i)$, which is finite because $f_i$ is bounded and the feasible set is nonempty.
Choose a maximizing sequence $(\mu_i^m)$ such that
\[
\int_K f_i\,d\mu_i^m\longrightarrow s.
\]
By sequential compactness, a subsequence satisfies
$\mu_i^{m_\ell}\xrightarrow{w^*}\mu_i^*
\in\mathcal F_i(z_i)$.
Since $f_i$ is bounded and upper semicontinuous, the Portmanteau inequality yields
\[
s
=\lim_{\ell\to\infty}\int_K f_i\,d\mu_i^{m_\ell}
\leq\int_K f_i\,d\mu_i^*
\leq s.
\]
Thus $\mu_i^*$ attains the primal maximum.

For strong duality, use $\pi(\tau)=1-\sum_{k<\tau}\pi(k)$ to replace the constraints by the equivalent normalized moments
\[
1,\ \pi(1),\ldots,\pi(\tau-1),\ v(\pi),
\]
with targets $(1,p_1,\ldots,p_{\tau-1},z_i)$.
The domain is compact metric, the moment functions are continuous, $f_i$ is upper semicontinuous, and the primal problem is feasible.
Corollary~3.1 of \citet{shapiro2001duality} therefore gives equality of the primal and dual values for this normalized moment problem.

Its dual majorant has the form $c_0+\sum_{k<\tau}c_k\pi(k)+\beta v(\pi)$.
The invertible substitution
\[
c_0=\alpha_\tau,\qquad
c_k=\alpha_k-\alpha_\tau\quad(k<\tau)
\]
transforms this majorant into $\alpha\cdot\pi+\beta v(\pi)$ and its objective into $\alpha\cdot p+\beta z_i$.
Thus the normalized dual is exactly $(D_i)$, proving $\operatorname{val}(P_i)=\operatorname{val}(D_i)$.
\end{proof}
\section{Certified interval and global payoff bounds}\label{app:certified-interval}
\begin{lemma}[Certified interval and global payoff bounds]
\label{lem:interval-payoff-bound}
Fix a nonmediator $i$, opponents' experiments $\mu_{-i}$, and a retained share $\lambda$. Write
\[
v(\pi):=\operatorname{Var}_\pi(X),
\qquad
g_z(\pi):=
f_i^\lambda(\pi;\mu_{-i},\bar\eta(z)).
\]
Let $I=[\ell,u]\subseteq\mathcal Z_i(p_{\mathcal E})$ be nonempty. Suppose a finite, nonempty family $\mathcal C_{i,I}$ of triples $(\alpha_c,\beta_c,d_c)\in\mathbb R^\tau\times\mathbb R^2$ satisfies
\begin{equation}\label{eq:interval-residual-certificate}
d_c\ge
\sup_{\substack{z\in I\\\pi\in K_\varepsilon}}
\bigl[
g_z(\pi)-\alpha_c\cdot\pi-\beta_cv(\pi)
\bigr],
\qquad c\in\mathcal C_{i,I}.
\end{equation}
Set $a_c:=\alpha_c\cdot p_{\mathcal E}+d_c$ and define
\[
B_i(I)
:=\max_{z\in I}
\min_{c\in\mathcal C_{i,I}}(a_c+\beta_cz).
\]
This maximum is finite and attained, and equivalently
\[
\begin{aligned}
B_i(I)=\max_{z,t\in\mathbb R}\quad&t\\
\text{subject to}\quad&
\ell\le z\le u,\\
&t\le a_c+\beta_cz
\qquad(c\in\mathcal C_{i,I}).
\end{aligned}
\]
Any verified numerical upper bound
$\overline B_i^{\mathrm{num}}(I)\ge B_i(I)$ satisfies
\[
\sup_{z\in I}W_i(z;\mu_{-i})
\le B_i(I)
\le\overline B_i^{\mathrm{num}}(I).
\]

In particular, choose verified weights
$\omega_c\ge0$ with $\sum_c\omega_c=1$, and set
\[
A_\omega:=\sum_c\omega_ca_c,
\qquad
b_\omega:=\sum_c\omega_c\beta_c.
\]
A valid numerical bound is
\[
\overline B_i^{\mathrm{num}}(I)
:=
\operatorname{up}\!
\left(A_\omega+\max\{\ell b_\omega,u b_\omega\}\right),
\]
where $\operatorname{up}$ denotes evaluation with a
verified upper enclosure of all numerical errors.

Finally, if a finite partition $\mathcal P_i$ covers
$\mathcal Z_i(p_{\mathcal E})$, then
\[
W_i^\star(\mu_{-i})
\le
\overline W_i^{\mathrm{num}}
:=
\max_{I\in\mathcal P_i}
\overline B_i^{\mathrm{num}}(I).
\]
\end{lemma}

\begin{proof}
Fix $z\in I$ and an admissible experiment $\nu_i$ with $R_i(\nu_i)=z$.
For every $c\in\mathcal C_{i,I}$,
\eqref{eq:interval-residual-certificate} and the moment constraints give
\[
\begin{aligned}
\int_{K_\varepsilon}g_z(\pi)\,\nu_i(d\pi)
&\le
\int_{K_\varepsilon}
\bigl[\alpha_c\cdot\pi+\beta_cv(\pi)+d_c\bigr]
\,\nu_i(d\pi)\\
&=\alpha_c\cdot p_{\mathcal E}+\beta_cz+d_c\\
&=a_c+\beta_cz.
\end{aligned}
\]
Here $\nu_i(K_\varepsilon)=1$.
Taking the supremum over such experiments yields
\[
W_i(z;\mu_{-i})
\le
\min_{c\in\mathcal C_{i,I}}(a_c+\beta_cz).
\]
Maximizing over $z\in I$ gives the interval bound.

For each fixed $z$, the largest feasible $t$ in the displayed LP is precisely $\min_c(a_c+\beta_cz)$, proving the LP representation.
This minimum is a finite continuous function of $z$, so its maximum on the compact interval $I$ is attained.

For any admissible weights $\omega$,
\[
\min_c(a_c+\beta_cz)
\le
\sum_c\omega_c(a_c+\beta_cz)
=A_\omega+b_\omega z.
\]
Therefore
\[
B_i(I)
\le A_\omega+\max\{\ell b_\omega,u b_\omega\}
\le\overline B_i^{\mathrm{num}}(I).
\]
Taking the maximum over the covering partition proves the global bound.
\end{proof}
\end{appendices}
\end{document}